\documentclass[a4paper]{article}

\usepackage{amsmath}
\usepackage{amssymb}
\usepackage{graphicx}
\usepackage[OT4]{fontenc}
\usepackage[latin2]{inputenc}
\usepackage{tikz}
\usepackage{url}
\usepackage[numbers]{natbib}
\usepackage{bm}
\usepackage{setspace}

\newcommand{\okra}[1]{\left( #1 \right)}
\newcommand{\kwad}[1]{\left[ #1 \right]}
\newcommand{\klam}[1]{\left\{ #1 \right\}}

\DeclareMathOperator{\sred}{\mathbf{E}}

\newcommand{\boole}[1]{{\bf 1}{\klam{#1}}}

\newtheorem{definition}{Definition}

\newtheorem{theorem}{Theorem}
\newtheorem{lemma}{Lemma}
\newenvironment*{proof}{\begin{trivlist}\item[]
\noindent\textbf{Proof:}}{$\Box$\par\end{trivlist}}
\newenvironment*{proof*}[1]{\begin{trivlist}\item[]
\noindent\textbf{Proof of #1:}}{$\Box$\par\end{trivlist}}

\author{{\L}ukasz D\k{e}bowski\thanks{
    {\L}. D\k{e}bowski is with
    the Institute of Computer Science, Polish Academy of Sciences, 
    ul. Jana Kazimierza 5, 01-248 Warszawa, Poland 
    (e-mail: ldebowsk@ipipan.waw.pl).}}

\title{A Preadapted Universal Switch Distribution for Testing
  Hilberg's Conjecture} \date{}

\begin{document}

\pagestyle{empty}   
\begin{titlepage}
\maketitle

\begin{abstract}
  Hilberg's conjecture about natural language states that the mutual
  information between two adjacent long blocks of text grows like a
  power of the block length. The exponent in this statement can be
  upper bounded using the pointwise mutual information estimate
  computed for a carefully chosen code. The bound is the better, the
  lower the compression rate is but there is a requirement that the
  code be universal.  So as to improve a received upper bound for
  Hilberg's exponent, in this paper, we introduce two novel universal
  codes, called the plain switch distribution and the preadapted
  switch distribution. Generally speaking, switch distributions are
  certain mixtures of adaptive Markov chains of varying orders with
  some additional communication to avoid so called catch-up
  phenomenon. The advantage of these distributions is that they both
  achieve a low compression rate and are guaranteed to be
  universal. Using the switch distributions we obtain that a sample of
  a text in English is non-Markovian with Hilberg's exponent being
  $\le 0.83$, which improves over the previous bound $\le 0.94$
  obtained using the Lempel-Ziv code.
  \\[1em]
\end{abstract}


\end{titlepage}
\pagestyle{plain}   


\section{Introduction}


Hilberg's conjecture is a hypothesis concerning natural language which
states that the mutual information between two adjacent long blocks of
text grows very fast, namely as a power of the block length
\cite{Hilberg90, EbelingNicolis91, EbelingNicolis92, EbelingPoschel94,
  BialekNemenmanTishby01, BialekNemenmanTishby01b,
  CrutchfieldFeldman03}. There are two important information-theoretic
results concerning this conjecture. On the one hand, Hilberg's
hypothesis can be linked with the idea that texts in natural language
refer to large amounts of randomly accessed knowledge in a repetitive
way \cite{Debowski11b,Debowski12}. On the other hand, Hilberg's
hypothesis can be linked with the fact that the number of distinct
words in a text grows as a power of the text length
\cite{Debowski06,Debowski11b}, the fact known as Herdan's or Heaps'
law \cite{Herdan64,Heaps78}. These two results make Hilberg's
conjecture interesting and worth direct empirical testing.

To present Hilberg's conjecture formally, let us introduce some
notations.  Consider a probability space $(\Omega,\mathcal{J},Q)$ with
$\Omega=\klam{1,2,...,D}^{\mathbb{Z}}$, random variables
$X_k:\Omega\ni(x_i)_{i\in\mathbb{Z}}\mapsto x_k\in\klam{1,2,...,D}$,
and distribution $Q$ which is stationary on $(X_i)_{i\in\mathbb{Z}}$
but not necessarily ergodic.  Blocks of symbols or variables are
denoted as $X_n^m=(X_i)_{n\le i\le m}$ with $X_n^m$ being the empty
block for $m<n$.  Moreover, for a random variable $X$ we introduce a
random variable $Q(X)$ which takes value $Q(X=x)$ for $X=x$. The
pointwise entropy of variable $X$ is the random variable
\begin{align}
  H^Q(X)=-\log Q(X)
\end{align}
whereas the pointwise mutual information between $X$ and $Y$ is
\begin{align}
  I^Q(X;Y)=-\log Q(X)-\log Q(Y)+\log Q(X,Y)
  .
\end{align}
Having this in mind, Hilberg's conjecture states that
\begin{align}
  \label{Hilberg}
  I^Q(X_1^n;X_{n+1}^{2n})&\propto n^\beta
  ,
  \quad
  \beta\in(0,1)
  .
\end{align}
Hilberg \cite{Hilberg90} supposed that $\beta\approx 0.5$ holds for
texts in English but his estimate was very rough, based on Shannon's
psycholinguistic experiment \cite{Shannon51}.

It is an interesting open question how much the exponent $\beta$
varies across different texts and whether it is possibly a
text-independent language universal. This question can be connected to
some fundamental limitations of human memory and
attention. Consequently, in this paper, we want to improve the method
of upper bounding Hilberg's exponent $\beta$ proposed in
\cite{Debowski13}, which is based on universal coding
\cite{ZivLempel77} or universal distributions \cite{Ryabko10}. The
focus of the present paper is to develop a better tool for estimating
mutual information than used so far and to demonstrate how it works
with an instance of empirical language data. Whereas we propose some
method of upper bounding exponent $\beta$ and showing that texts in
natural language are non-Markovian, it remains to find a computational
method of lower bounding Hilberg's exponent $\beta$.

Before we show how to upper bound exponent $\beta$ using a universal
distribution, it is important to note that the discussion of Hilberg's
hypothesis is intimately connected with the question whether the
natural language production is nonergodic and what its plausible
ergodic decomposition is.  According to the ergodic decomposition
theorems \cite{GrayDavisson74b,Kallenberg97,Debowski11b}, any
stationary measure $Q$ equals the expectation $\sred_Q F$, where
$F=Q(\cdot|\mathcal{I})$ is the random ergodic measure for measure $Q$
and $\mathcal{I}$ is the shift-invariant algebra.  There exist some
stationary nonergodic measures $Q$, called Santa Fe processes, for
which mutual information $I^Q(X_1^n;X_{n+1}^{2n})$ grows according to
a power law but the main contribution of the mutual information comes
from identifying the random ergodic measure $F$ given the block
$X_1^n$ \cite{Debowski11b,Debowski12}. In that case mutual information
$I^F(X_1^n;X_{n+1}^{2n})$ for the random ergodic measure $F$ itself is
negligibly small. The Santa Fe processes are not irrelevant for our
discussion. In their original construction they were intended as some
idealized models for the transmission of knowledge in natural
language. Thus when estimating the exponent in Hilberg's conjecture we
have first to decide whether we do it for measure $Q$ or for measure
$F$. The methods for estimating $I^Q(X_1^n;X_{n+1}^{2n})$ and
$I^F(X_1^n;X_{n+1}^{2n})$ are very different.  We suppose that
estimating the mutual information for the nonergodic measure $Q$ is
closer to the original intention of Hilberg although some
philosophical problem remains `how many' ergodic components we
actually admit (e.g.\ do we assume that $Q$ models texts in a given
register of a particular language or in any register of any natural
language).  In contrast to the random ergodic measure $F$, the
possibly nonergodic measure $Q$ is not identifiable given a single
realization $(X_i)_{i\in\mathbb{Z}}$. Despite that, it is somewhat
baffling that some nontrivial upper bound for Hilberg's exponent
$\beta$, a~property of measure $Q$, can be learned from a single
realization $(X_i)_{i\in\mathbb{Z}}$ if the growth of mutual
information is uniform in $Q$.

As we have indicated, our method of upper bounding Hilberg's exponent
$\beta$ is based on universal coding.  Here we say that a~distribution
$P$ is weakly universal if for every stationary distribution $Q$ we
have
\begin{align}
  \label{WeaklyUniversal}
  \lim_{n\rightarrow\infty} \frac{1}{n} \sred_Q H^P(X_1^n)
  = h_Q
  ,
\end{align}
where the entropy rate $h_Q$ is
\begin{align}
  \label{EntropyRate}
  h_Q
  :=
  \lim_{n\rightarrow\infty} \frac{1}{n} \sred_Q H^Q(X_1^n)
  =
  \inf_{k\in\mathbb{N}} \sred_Q \kwad{-\log Q(X_{k+1}|X_1^k)}
  .
\end{align}
On the other hand, the distribution $P$ is called strongly universal
if for every stationary ergodic distribution $Q$ we have $Q$-almost
surely
\begin{align}
  \label{StronglyUniversal}
  \limsup_{n\rightarrow\infty} \frac{1}{n} H^P(X_1^n)
  \le
  h_Q
  .
\end{align}
Strongly universal distributions are weakly universal under mild
conditions \cite{Weissman05}.
Dębowski \cite{Debowski13} proposed to investigate the empirical law
of form
\begin{align}
  \label{HilbergP}
  I^P(X_1^n;X_{n+1}^{2n})&\propto n^{\gamma}
  ,
  \quad
  \gamma\in(0,1)
  ,
\end{align}
where $P$ is a strongly universal distribution. Relationship
(\ref{HilbergP}), which can be called the codewise Hilberg conjecture,
has been checked experimentally for the Lempel-Ziv code on a sample of
10 texts in English and it holds surprisingly uniformly with
$\gamma=0.94$ \cite{Debowski13}. The same estimate $\gamma=0.94$ has
been obtained for 21 other texts in German and French (work under
review).

Are laws (\ref{Hilberg}) and (\ref{HilbergP}) related?  In fact, if
they hold uniformly for large $n$ then exponents $\beta$ and $\gamma$
can be linked.  To see it, the following lemma is helpful.
\begin{lemma}[\mbox{\cite{Debowski11b}}]
\label{theoExcessBound}
Consider a~function $G:\mathbb{N}\rightarrow\mathbb{R}$ such that
$\lim_k G(k)/k=0$ and $G(n)\ge 0$ for all but finitely many $n$. For
infinitely many $n$, we have $2G(n)- G(2n)\ge 0$.
\end{lemma}
If $P$ is weakly universal, the above statement is satisfied for
Kullback-Leibler divergence 
\begin{align}
  G(n)=\sred_Q\kwad{H^P(X_1^n)-H^Q(X_1^n)}
  .
\end{align}
Hence we obtain that
\begin{align}
  \label{MCIMI}
  \sred_Q I^P(X_1^n;X_{n+1}^{2n})
  \ge
  \sred_Q I^Q(X_1^n;X_{n+1}^{2n})
\end{align}
holds for infinitely many $n$. Thus if relationships (\ref{Hilberg})
and (\ref{HilbergP}) hold uniformly for large $n$  then
\begin{align}
 \label{GammaBeta}
 \gamma\ge \beta 
 .
\end{align}
In other words, the smaller $\gamma$ we observe for a text (or
methodologically better, for a large sample of different texts), the
better bound it gives for $\beta$.  It can also be easily shown that
the bound is the tighter, the smaller compression rate $H^P(X_1^n)/n$
is, with the sole provision that distribution $P$ be weakly
universal. Results of our experiment suggest that this requirement is
essential.

Thus the question of upper bounding Hilberg's exponent $\beta$ boils
down, if we assume uniform information growth in (\ref{Hilberg}), to
finding appropriate universal distributions. Many methods have been
proposed for compression of texts in natural language, e.g.:
Lempel-Ziv (LZ) code \cite{ZivLempel77}, $n$-gram models
\cite{BrownOthers92,Jelinek97,ManningSchutze99}, prediction by partial
match (PPM) \cite{ClearyWitten84}, context tree weighting (CTW)
\cite{WillemsShtarkovTjalkens95}, probabilistic suffix trees (PST)
\cite{RonSingerTishby96}, grammar-based codes \cite{KiefferYang00},
PAQ codes \cite{SalomonMotta}, and switch distributions
\cite{ErvenGrunwaldRooij07}. These compression schemes can be divided
into two classes: (a) preadapted distributions, which are trained on
large corpora and achieve low compression rate---as low as 0.88 bpc
(bits per character) for WinRK
3.1.2,\footnote{\url{http://www.maximumcompression.com/data/text.php}}
and (b) adaptive distributions, which are not pre-trained and achieve
larger compression rate but are proven to be universal. Whereas the
distributions proposed so far belong either to class (a) or (b), for
upper bounding Hilberg's exponent, we need a distribution that would
combine the advantages of classes (a) and (b), namely low compression
rate and universality.

In this paper we propose and investigate two novel universal
distributions, one of which is not preadapted and the other is
preadapted.  The point of our departure is a modification of the
switch distributions proposed in
\cite{ErvenGrunwaldRooij07,KollenRooij13}.  The idea of a~switch
distribution is to use a mixture of adaptive Markov chains of varying
orders but, at each data point, the probabilities are partly
transferred among different orders. In this way, lower order Markov
chains are used to compress the data exclusively until enough
information is gathered to predict new outcomes with higher order
chains.  This avoids so called catch-up phenomenon and leads to much
better compression than while there is no transmission of
probabilities among different Markov chain orders
\cite{ErvenGrunwaldRooij07}. If we combine the idea of the switch
distribution with smoothing proposed in \cite[p. 111]{HindleRooth93}
and the idea of a universal distribution called R measure, proposed in
\cite{Ryabko10}, we obtain another new universal compression scheme,
which is efficiently computable.  This scheme will be called the plain
switch distribution. It is not preadapted yet. The preadapted switch
distribution is obtained by initializing the Markov chains with
frequencies coming from a large corpus and letting them gradually
adapt to the compressed source.  It will be shown that the preadapted
switch distributions is also universal. For the considered input text
the nonpreadapted and the preadapted switch distributions achieve
almost the same ultimate compression rate 2.21 bpc, approximately
twice smaller than for the LZ code. This figure is not so favorable as
for the WinRK 3.1.2 but we have a guarantee that the switch
distributions are universal.

Once we have constructed the universal switch distributions, we can
use them for upper bounding Hilberg's exponent. In the previous paper
\cite{Debowski13}, the LZ code was used for a sample of texts in
English which yielded $\gamma=0.94$. Here using the plain switch
distribution we obtain a slightly tighter bound $\gamma=
0.83$. Surprisingly, the preadapted switch distribution yields almost
the same compression rate for long blocks as the plain switch
distribution and does not give a tighter bound for
$\gamma$. Differences in the estimates of $\gamma$ may also stem from
differences in data representation. In \cite{Debowski13} the alphabet
of $D=27$ symbols was used. Here we use $D=256$ and obtain $\gamma=
0.89$ for the LZ code. It is important to underline that meaningful
estimates of $\gamma$ can be only obtained using universal
distributions. As we show, if a nonuniversal distribution is used, the
pointwise mutual information can be very low despite a good-looking
compression rate. To a certain extent this also applies to the
preadapted switch distribution, where the pointwise mutual information
is low for short blocks.

There remains a question what the estimates of the codewise Hilberg
exponent $\gamma$ tell about the true Hilberg exponent $\beta$ for
texts in natural language. In particular, how large can the difference
between $\gamma$ and $\beta$ be in case of the considered universal
codes?  Let us recall that for memoryless sources, i.e., IID processes
with $\beta=0$, the pointwise mutual information of the LZ code is of
order $I^P(X_1^n;X_{n+1}^{2n})\propto n/\log n$
\cite{LouchardSzpankowski97} so empirically we should observe
$\gamma\approx 1$. Thus the difference between between $\gamma$ and
$\beta$ can be close to $1$. Moreover, observing $\gamma\approx 1$ for
the LZ code we cannot reject the hypothesis that the source is IID. In
contrast, for stationary Markov chains, the pointwise mutual
information of a Bayesian mixture of Markov chains, which is the
building block for the switch distribution, cf. \cite{Grunwald07}, is
only of order $I^P(X_1^n;X_{n+1}^{2n})\propto \log n$ \cite{Atteson99}
so empirically we should observe $\gamma\approx 0$ in that case. The
same property carries over to the switch distributions introduced in
this paper.  Thus, observing $\gamma> 0$ for a switch distribution,
we are compelled to reject the hypothesis that the source is a Markov
chain. Possibly, $\beta>0$ may hold in that case. In other words,
given the presented empirical data, Hilberg's conjecture may be true
but we need still some stronger evidence in favor of this hypothesis,
such as a nontrivial lower bound for exponent $\beta$ (work under
review).

The further organization of the paper is as follows.  In Section
\ref{secSwitch}, we present the plain switch distribution. In Section
\ref{secAdapt}, we discuss the preadapted switch distribution.  In
Section \ref{secExperiments}, we investigate the codewise Hilberg's
conjecture (\ref{HilbergP}) experimentally using the introduced
distributions.

\section{The plain switch distribution}
\label{secSwitch}

The frequency of substring $w_1^k\in\klam{1,2,...,D}^k$ in string
$z_1^n\in\klam{1,2,...,D}^n$ will be denoted as
\begin{align}
  c(w_1^k|z_1^n)=\sum_{i=0}^{n-k}\boole{w_1^k=z_{i+1}^{i+k}}
  .
\end{align}
The plain switch distribution is defined as follows:
\begin{definition}[plain switch distribution]
  Define conditional probabilities $B(x_{n+1}|x_1^n,-1)=D^{-1}$ and
  \begin{align}
  \label{kBayesianII}
  B(x_{n+1}|x_1^n,k)
  &=
  \frac{
    c(x_{n+1-k}^{n+1}|x_1^n) +B(x_{n+1}|x_1^n,k-1)
  }{
    c(x_{n+1-k}^{n}|x_1^{n-1}) +1
  }
  .
\end{align}
Let coefficients $p_n\in(0,1)$, where $n=0,1,2,...$, satisfy
$\prod_{n=0}^{\infty} p_n>0$. Put also $q_n=1-p_n$. We define the
partial switch distribution $P(x_1^n,k)$ by conditions
\begin{align}
  \label{SwitchI}
  P(x_1,-1)&=p_0 B(x_1|-1)
  ,
  \\
  \label{SwitchII}
  P(x_1,0)&=q_0 B(x_1|0)
  ,
  \\
  \label{SwitchIII}
  P(x_1^n,k)&=0
  \text{ for $k<-1$ or $k\ge n$}
  ,
  \\
  P(x_1^{n+1},k)&=\kwad{
    p_n P(x_1^n,k)
    +q_n P(x_1^n,k-1)
  } B(x_{n+1}|x_1^n,k)
  \nonumber
  \\
  \label{SwitchIV}
  &\qquad\qquad\qquad\qquad\qquad\qquad\qquad
  \text{for $n\ge 1$ and $-1\le k\le n$}
  .
\end{align}
The total probability for block $x_1^n$ according to the switch
distribution is
\begin{align}
  \label{TotalSwitch}
  P(x_1^n)=
  \sum_{k=-1}^{n-1} P(x_1^n,k)
  .
\end{align}
The scheme of computing $P(x_1^n)$ is depicted in Figure
\ref{figSwitch}.
\end{definition}

\emph{Remark 1:} Condition $\prod_{n=0}^{\infty} p_n>0$ holds for
instance if we fix
\begin{align}
  \label{pn}
  p_n=\exp\kwad{-(n+1)^{-\alpha}}, \quad \alpha>1
  .
\end{align}
Value $\alpha$ is a parameter.

\emph{Remark 2:} Probability $B(x_{n+1}|x_1^n,k)$ defines an adaptive
$k$-th order Markov model.  Probability $P(x_1^n,k)$ represents the
mass of the adaptive $k$-th order Markov model modified by
communication with models of lower orders.  The motivation for this
communication, carried out in formula (\ref{SwitchIV}), is that lower
order Markov models should be solely used for compression until enough
data are collected to predict new outcomes with higher order Markov
models, cf., the catch-up phenomenon described in
\cite{ErvenGrunwaldRooij07}. Distribution $P(x_1^n)$ is a~special case
of the general scheme of switch distributions considered by
\cite{ErvenGrunwaldRooij07,KollenRooij13} to overcome the catch-up
phenomenon.  In contrast to the models discussed in
\cite{ErvenGrunwaldRooij07,KollenRooij13}, the switch distribution
considered here is universal in the sense made precise in the
Introduction and still can be efficiently computed.

\emph{Remark 3:} Probability $B(x_{n+1}|x_1^n,k)$ of an adaptive
$k$-th order Markov model so as to be defined for all $x_1^n$ is
smoothed using probability $B(x_{n+1}|x_1^n,k-1)$ of an adaptive
$(k-1)$-th order Markov model in a way inspired by
\cite[p. 111]{HindleRooth93}. Thus we add the probability of a lower
order $B(x_{n+1}|x_1^n,k-1)$ in the numerator and $1$ in the
denominator rather than using the Laplace rule or the
Krichevski-Trofimov rule, i.e., adding $\alpha\in(0,1]$ in the
numerator and $\alpha D$ in the denominator as in, e.g.,
\cite{Ryabko10}. We have checked that this trick works better for
natural language data than the Laplace or Krichevski-Trofimov rules.

\emph{Remark 4:} The infinite mixture of probabilities
$B(x_{n+1}|x_1^n,k)$ for orders $k=0,1,2,...$, smoothed using the
Laplace or Krichevski-Trofimov rules, was investigated in
\cite{Ryabko10} under the name of R measure and shown to be a
universal distribution. Our construction is quite similar in spirit
but avoids the catch-up phenomenon.

\begin{figure}[t]
  \centering
\includegraphics[width=0.7\textwidth]{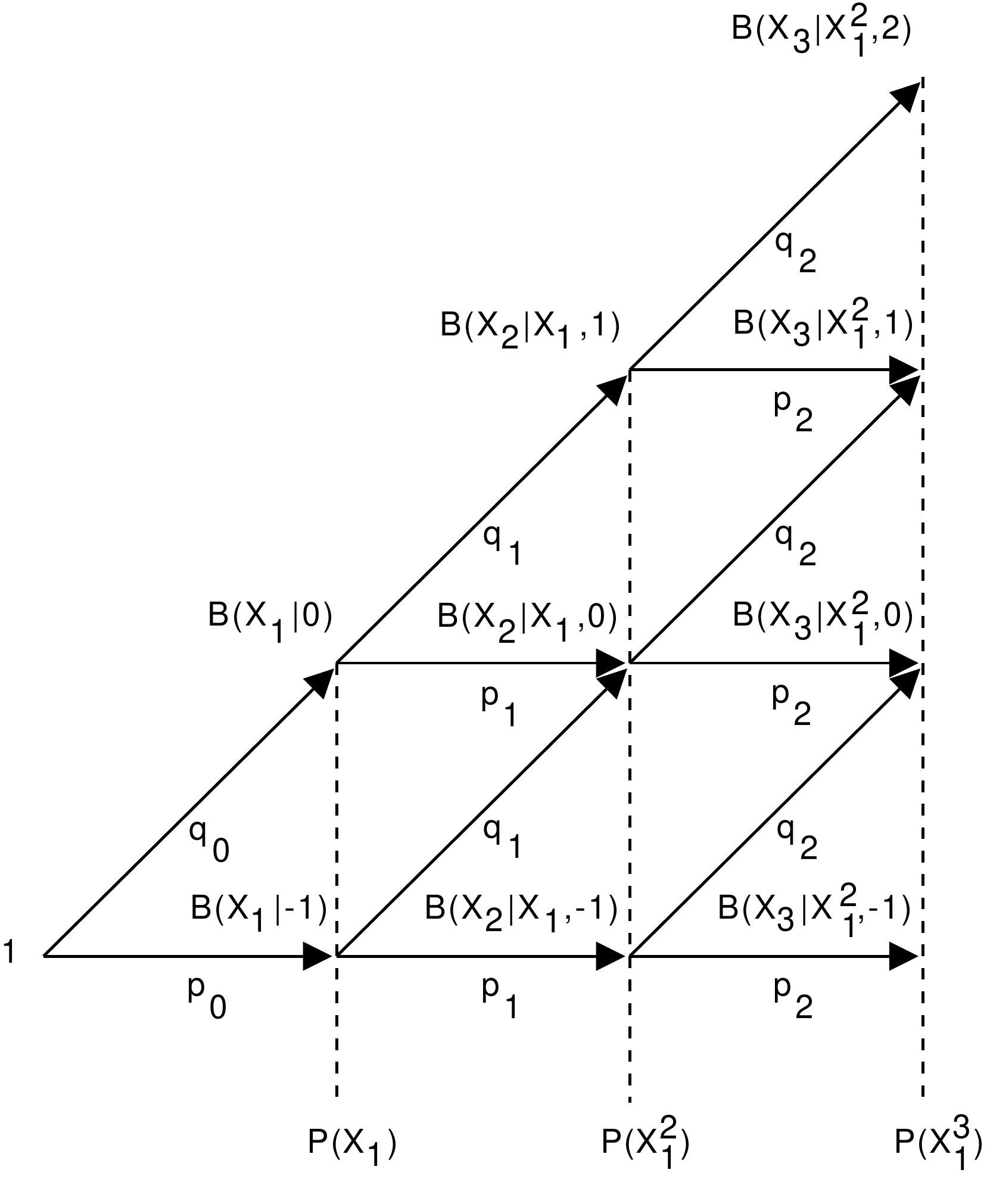}
  \caption{\label{figSwitch}
    The scheme of computing  $P(x_1^3)$.
}
\end{figure}

Now we will show that the switch distribution (\ref{TotalSwitch}) is
both strongly and weakly universal. First we need this simple fact:
\begin{lemma}
  \label{theoDomination}
  Introduce notation
  \begin{align}
    \label{TotalB}
    B(x_l^n|x_{l-k}^{l-1},k)=\prod_{i=l}^n B(x_{i}|x_1^{i-1},k)
    .
  \end{align}
  The switch distribution satisfies the following:
  \begin{enumerate}
  \item[i)] there exists a~constant $\delta_{-1}>0$ such that for all
    $n\ge 1$ we have
    \begin{align}
      \label{BayesianDomI}
      P(x_1^n)&\ge \delta_{-1} D^{-n}
      ,
    \end{align}
  \item[ii)] for each $k\ge 0$ there exists a~constant $\delta_k>0$
    such that for all $n\ge k+1$ we have
    \begin{align}
    \label{BayesianDomII}
    P(x_1^n)&\ge \delta_k B(x_{k+1}^n|x_1^{k},k)
    .
  \end{align}
\end{enumerate}  
\end{lemma}
\begin{proof}
  For $n\ge 1$ we have
\begin{align}
  P(x_1^n)&\ge 
  \okra{\prod_{i=0}^{n-1} p_i B(x_{i+1}|x_1^i,-1)} 
  \ge \delta_{-1} P(x_1^n|-1)
  .
\end{align}
where $\delta_{-1}=\prod_{i=0}^{\infty} p_i>0$.  Thus we have claim
(i). On the other hand, for $k\ge 0$ and $n\ge k+1$ we obtain
\begin{align}
  P(x_1^n)&\ge 
  \okra{\prod_{i=0}^{k} q_i B(x_{i+1}|x_1^i,i)} 
  \okra{\prod_{i=k+1}^{n-1} p_i B(x_{i+1}|x_1^i,k)} 
  \nonumber
  \\
  &=
  \okra{\prod_{i=0}^{k} q_i} D^{-k}
  \okra{\prod_{i=k+1}^{n-1} p_i}
  B(x_{k+1}^n|x_1^{k},k)
  \nonumber
  \\
  &\ge
  \delta_k B(x_{k+1}^n|x_1^{k},k)
  ,
\end{align}
where 
\begin{align}
  \delta_k =   
  \okra{\prod_{i=0}^{k} q_i} D^{-k}
  \okra{\prod_{i=k+1}^{\infty} p_i}
  >0
  .
\end{align}
Hence the claim (ii) follows.  
\end{proof}

Combining Lemma \ref{theoDomination}(ii) with the ergodic theorem we
obtain the proof of universality.
\begin{theorem}
  \label{theoUniversal}
  The switch distribution is strongly and weakly universal.
\end{theorem}
\begin{proof}
  Let $Q$ be a stationary ergodic distribution.  Since the alphabet of
  $X_i$ is finite, by the ergodic theorem differences
  $B(X_{n}|X_1^{n-1},k)-Q(X_{n}|X_{n-k-1}^{n-1})$ converge to $0$
  $Q$-almost surely. Hence
  \begin{align}
    \lim_{n\rightarrow\infty}\frac{1}{n} \kwad{-\log
      B(X_{k+1}^n|X_1^k,k)} = \lim_{n\rightarrow\infty}\frac{1}{n}
    \kwad{-\sum_{i=k+1}^n \log Q(X_{i}|X_{i-k-1}^{i-1})}
    .
  \end{align}
  Applying the ergodic theorem again, we obtain
  \begin{align}
    \lim_{n\rightarrow\infty}\frac{1}{n}
    \kwad{-\sum_{i=k+1}^n \log Q(X_{i}|X_{i-k-1}^{i-1})}
    =\sred_Q \kwad{-\log Q(X_{k+1}|X_1^k)}
    .
  \end{align}

  Now we combine these facts with Lemma \ref{theoDomination}(ii),
  which yields
  \begin{align}
    \limsup_{n\rightarrow\infty} \frac{1}{n} \kwad{-\log P(X_1^n)}
    &\le
    \inf_{k\in\mathbb{N}} \lim_{n\rightarrow\infty}\frac{1}{n} \kwad{-\log
      B(X_{k+1}^n|X_1^k,k)}
    \\
    &=
    \inf_{k\in\mathbb{N}} \sred_Q \kwad{-\log Q(X_{k+1}|X_1^k)}=h_Q
    .
  \end{align}
  Hence the distribution $P$ is strongly universal.  Moreover, as
  shown in \cite{Weissman05}, the claim of Lemma
  \ref{theoDomination}(i) and the strong universality are sufficient
  conditions that distribution $P$ be weakly universal.
\end{proof}

A~naive implementation of the switch distribution $P(x_1^n)$ has the
time complexity $O(n^4)$ for the following reason: There are $O(n^2)$
calls of $B(x_{l+1}|x_1^l,k)$ where $k,l\le n$ and in a naive
implementation each $B(x_{l+1}|x_1^l,k)$ has time complexity
$O(kl)$. This is, however, a very careless approach and usually we can
do much better.  Let us denote the maximal length of a~substring that
appears at least twice in a string $z_1^n$ as
\begin{align}
\label{DefL}
  \mathbf{L}(z_1^n)&:=\max 
  \klam{k: \exists w_1^k: c(w_1^k|z_1^n)>1}
  .
\end{align}
For brevity, $\mathbf{L}(z_1^n)$ will be called the depth of $z_1^n$.
\begin{theorem}
  \label{theoComputable}
  The value of the switch distribution $P(x_1^n)$ can be computed in
  time $O(ns)$ where $s=\mathbf{L}(x_1^n)$ is the depth of $x_1^n$.
\end{theorem}
\emph{Remark:} The depth $\mathbf{L}(X_1^n)$ is bounded by $O(\log n)$
for a~large class of processes called finite-energy processes. They
can be obtained by dithering ergodic processes with an IID noise
\cite{Shields97}. For texts in natural language, an experiment indicates
that the depth $\mathbf{L}(x_1^n)$ is of order $O((\log n)^\alpha)$,
where $\alpha<4$ \cite{Debowski12b}.
\begin{proof}
  We can knock down the complexity of an individual call of
  $B(x_{l+1}|x_1^l,k)$ to a~constant if we store the frequencies of
  substrings tested in formula (\ref{kBayesianII}) and we increment
  them on line. Some further important savings can be done if we know
  the depth $s=\mathbf{L}(x_1^n)$. The value of $s$ can be computed in
  time $O(n)$ by building the suffix tree of $x_1^n$
  \cite{Ukkonen95}. Once we have that $s$, let us observe that
\begin{align}
  B(x_{l+1}|x_1^l,k)=B(x_{l+1}|x_1^l,s)
\end{align}
holds for all $k>s$.  

Thus we can flush all probabilities $P(x_1^n,k)$ for $k>s$ into
a~dummy variable $P(x_1^n,\bullet)$.  In the following, without
affecting the value of $P(x_1^n)$, the recursion
(\ref{SwitchIII})--(\ref{SwitchIV}) can be altered to
\begin{align}
  \label{SwitchIIIA}
  P(x_1^n,k)&=0
  \text{ for $k<-1$ or $k\ge n$}
  ,
  \\
  \label{SwitchIIIB}
  P(x_1^n,\bullet)&=0
  \text{ for $n< s+1$}
  ,
  \\
    P(x_1^{n+1},k)&=\kwad{
    p_n P(x_1^n,k)
    +q_n P(x_1^n,k-1)
  } B(x_{n+1}|x_1^n,k)
  \nonumber
  \\
  \label{SwitchIVA}
  &\qquad\qquad\qquad\qquad\qquad
  \text{for $n\ge 1$ and $-1\le k\le \min(n,s)$}
  ,
  \\
    P(x_1^{n+1},\bullet)&=\kwad{
    P(x_1^n,\bullet)
    +q_n P(x_1^n,s)
  } B(x_{n+1}|x_1^n,s)
  \label{SwitchIVB}
  \text{ for $n\ge s+1$}
  .
\end{align}
The formula for the total probability becomes
\begin{align}
  \label{TotalSwitchA}
  P(x_1^n)=
  \sum_{k=-1}^s P(x_1^n,k) + P(x_1^n,\bullet)
  .
\end{align}
Hence the time complexity of $P(x_1^n)$ is of order $O(ns)$.
\end{proof}

The space complexity of the switch distribution can also be reduced by
observing that in order to compute $B(x_{l+1}|x_1^l,k)$ we only need
to store the frequencies of substrings $w$ that appear in $x_1^n$ at
least twice and the frequencies of their extensions $wa$, where
$a\in\klam{1,2,...,D}$. These strings can be also found while
building the suffix tree of $x_1^n$.

Parameter $s$ in the algorithm
(\ref{SwitchIIIA})--(\ref{TotalSwitchA}) will be called the depth of
the switch distribution. Without a significant change of $P(x_1^n)$,
the depth of the switch distribution can be chosen as much smaller
than the depth of string $x_1^n$. This fact can be also used for the
further speed-up of computation.  Fixing the depth, however, leads
asymptotically to the $Q$-almost sure bound
\begin{align}
  \lim_{n\rightarrow\infty} \frac{1}{n} \kwad{-\log P(X_1^n)}
  &=
  \lim_{n\rightarrow\infty}\frac{1}{n} \kwad{-\log
    B(X_{s+1}^n|X_1^s,s)}
  \\
  &=
  \sred_Q \kwad{-\log Q(X_{s+1}|X_1^s)}
\end{align}
if $Q$ is ergodic.  Conditional entropy $\sred_Q \kwad{-\log
  Q(X_{s+1}|X_1^s)}$ is greater than $h_Q$.

\section{The preadapted switch distribution}
\label{secAdapt}

Often we want to predict or compress data $x_1^n$ that are generated
by a class of complex unknown distributions $Q$ that partly resemble
the empirical distribution of another, much larger data $y_1^j$. Such
a case arises in particular in the compression of texts in natural
language. Then using a universal distribution such as the plain switch
distribution need not be the best approach, since this distribution
has to learn all frequencies of substrings from the data $x_1^n$. A
competing approach is to use frequencies of substrings from the larger
data $y_1^j$. This can yield a better compression rate for finite data
$x_1^n$. The problem of using a fixed empirical distribution of
$y_1^j$ is, however, that it is not universal. The source of the
problem lies in using non-adaptive substring frequencies. A simple
solution for this problem is to initialize the substring frequencies
with the frequencies coming from $y_1^j$ and let them gradually adapt
to $x_1^n$. In this way we obtain a preadapted universal compression
scheme. One can suppose that this scheme may compress better than both
the plain switch distribution and the empirical distribution of
$y_1^j$.

Let us clarify this idea. 
\begin{definition}[fixed switch distribution]
  Let $y_1^j$ be a fixed sequence, called the training data.  Define
  conditional probabilities $B(x_{n+1}|x_1^n,-1)=D^{-1}$ and
  \begin{align}
    \label{kBayesianIIEmp}
    B(x_{n+1}|x_1^n,k)
    &=
    \frac{
      c(x_{n+1-k}^{n+1}|y_1^j) +B(x_{n+1}|x_1^n,k-1)
    }{
      c(x_{n+1-k}^{n}|y_1^{j-1}) +1
    }
    .
  \end{align}
  Using these $B(x_{n+1}|x_1^n,k)$, we define the fixed switch
  distribution $P$ via formulae (\ref{SwitchI})--(\ref{TotalSwitch}).
\end{definition}

For short blocks $x_1^n$, the fixed switch distribution can achieve
much lower compression rate than the plain switch distribution but it
is not universal.  To obtain a universal distribution which combines
the advantages of the fixed switch distribution and the plain switch
distribution, we may consider a~compromise between expressions
(\ref{kBayesianII}) and (\ref{kBayesianIIEmp}). This can be done
easily as follows.
\begin{definition}[preadapted switch distribution]
  Let $y_1^j$ be a fixed sequence, called the training data.
  Define conditional probabilities $B(x_{n+1}|x_1^n,-1)=D^{-1}$
  and
  \begin{align}
  B(x_{n+1}|x_1^n,k)=
  \label{kBayesianIIAdapt}
  &
  \frac{
    c(x_{n+1-k}^{n+1}|y_1^jx_1^n) +B(x_{n+1}|x_1^n,k-1)
  }{
    c(x_{n+1-k}^{n}|y_1^jx_1^{n-1}) +1    
  }
  .
\end{align}
Using these $B(x_{n+1}|x_1^n,k)$, we define the preadapted switch
distribution $P$ via formulae (\ref{SwitchI})--(\ref{TotalSwitch}).
\end{definition}

As in the plain case, we can show that the preadapted switch
distribution is universal and efficiently computable. The proof of
universality relies on the observation that the influence of training
data $y_1^j$ on the probability of long blocks $x_1^n$ is
asymptotically negligible.
\begin{theorem}
  \label{theoUniversalAdapt}
  The preadapted switch distribution is strongly and weakly
  universal.
\end{theorem}
\begin{proof}
  Analogously to the plain switch distribution, the preadapted switch
  distribution satisfies the analogue of  \ref{theoDomination}.
  Having this fact in mind, we can prove the universality.  Let $Q$ be
  a stationary ergodic distribution.  Since the alphabet of $X_i$ is
  finite, by the ergodic theorem differences
  $B(X_{n}|X_1^{n-1},k)-Q(X_{n}|X_{n-k-1}^{n-1})$ converge to $0$
  $Q$-almost surely. The further reasoning proceeds like the proof of
  Theorem \ref{theoUniversal}.
\end{proof}

\begin{theorem}
  \label{theoComputableAdapt}
  The value of the preadapted switch distribution $P(x_1^n)$ can be
  computed in time $O((j+n)s)$ where $s=\mathbf{L}(y_1^jx_1^n)$.
\end{theorem}
\begin{proof}
  The complexity of an individual call of $B(x_{l+1}|x_1^l,k)$
  can be reduced to a~constant if we record the frequencies of
  substrings tested in formula (\ref{kBayesianIIAdapt}) and we
  increment them on line. Initializing these frequencies takes time
  $O(js)$.  Let us also observe that
\begin{align}
  B(x_{l+1}|x_1^l,k)=B(x_{l+1}|x_1^l,s)
\end{align}
holds for all $k>s$.  Thus without affecting the value of $P(x_1^n)$,
the algorithm (\ref{SwitchIII})--(\ref{SwitchIV}) can be changed to
(\ref{SwitchIIIA})--(\ref{SwitchIVB}) and the formula for the total
probability becomes (\ref{TotalSwitchA}).  Thus the time complexity of
$P(x_1^n)$ is of order $O((j+n)s)$.
\end{proof}

The space complexity of the preadapted switch distribution can also be
reduced by noticing that in order to compute $B(x_{l+1}|x_1^l,k)$ we only
have to record the frequencies of substrings $w$ that appear in
$y_1^jx_1^n$ at least twice and the frequencies of their extensions
$wa$, where $a\in\klam{1,2,...,D}$.

\section{Measuring codewise Hilberg exponent $\gamma$}
\label{secExperiments}

Here we describe a~simple experiment that we have performed using the
three switch distributions and the Lempel-Ziv code.  As the training
data we have taken \emph{The Complete Memoirs} by J.\ Casanova
(6,719,801 characters), and as the compressed text---\emph{Gulliver's
  Travels} by J.\ Swift (579,438 characters). Both texts were
downloaded from the Project
Gutenberg.\footnote{\url{http://www.gutenberg.org/}} The alphabet size
was set as $D=256$. The switch distributions were computed using
transition probabilities $p_n$ of form (\ref{pn}) with $\alpha=1.001$
since we observed that the lower the $\alpha$ is the better
compression is achieved. Moreover, we have used algorithm
(\ref{SwitchIIIA})--(\ref{TotalSwitchA}) with fixed depth $s=7$ since
more than $99.99\%$ of the probability mass in the observed cases
concentrated in $P(x_1^n,k)$ with $k\le 4$.  Hence it was a safe
approximation.  The Lempel-Ziv code was computed by our own
implementation for the ASCII encoding of the text. The results are
presented in Tables \ref{tabCompressionRates} and \ref{tabPMI} and
Figures \ref{figCompressionRates} and \ref{figPMI}.

\begin{figure}[p]
  \centering
\includegraphics[width=\textwidth]{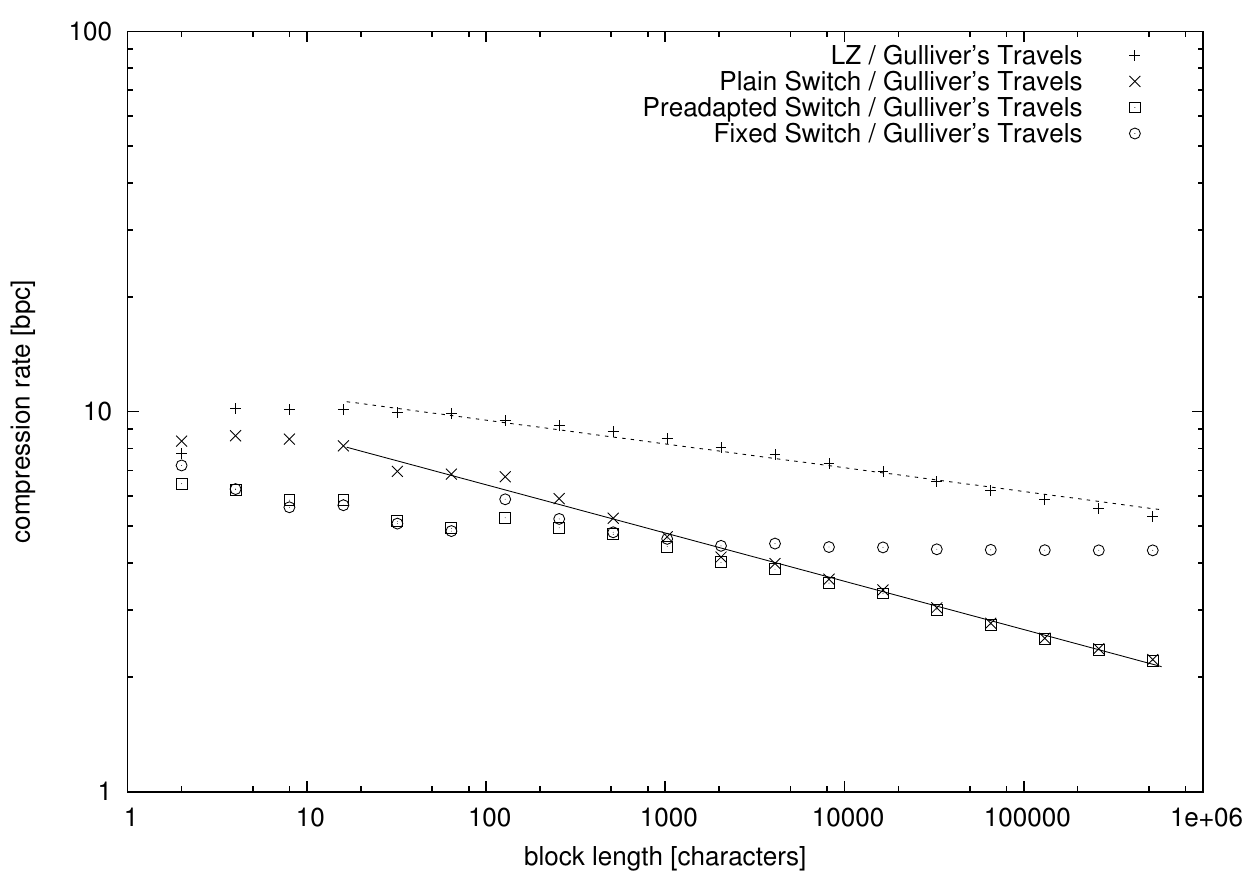}
\caption{\label{figCompressionRates} Compression rates for the switch
  distributions and the LZ code.  The solid line is the least square
  regression $y=11.51 n^{-0.127}$, computed for the plain switch
  distribution.  The dotted line is the least square regression
  $y=12.66 n^{-0.0625}$, computed for the LZ code. }
\end{figure}

\begin{figure}[p]
  \centering
\includegraphics[width=\textwidth]{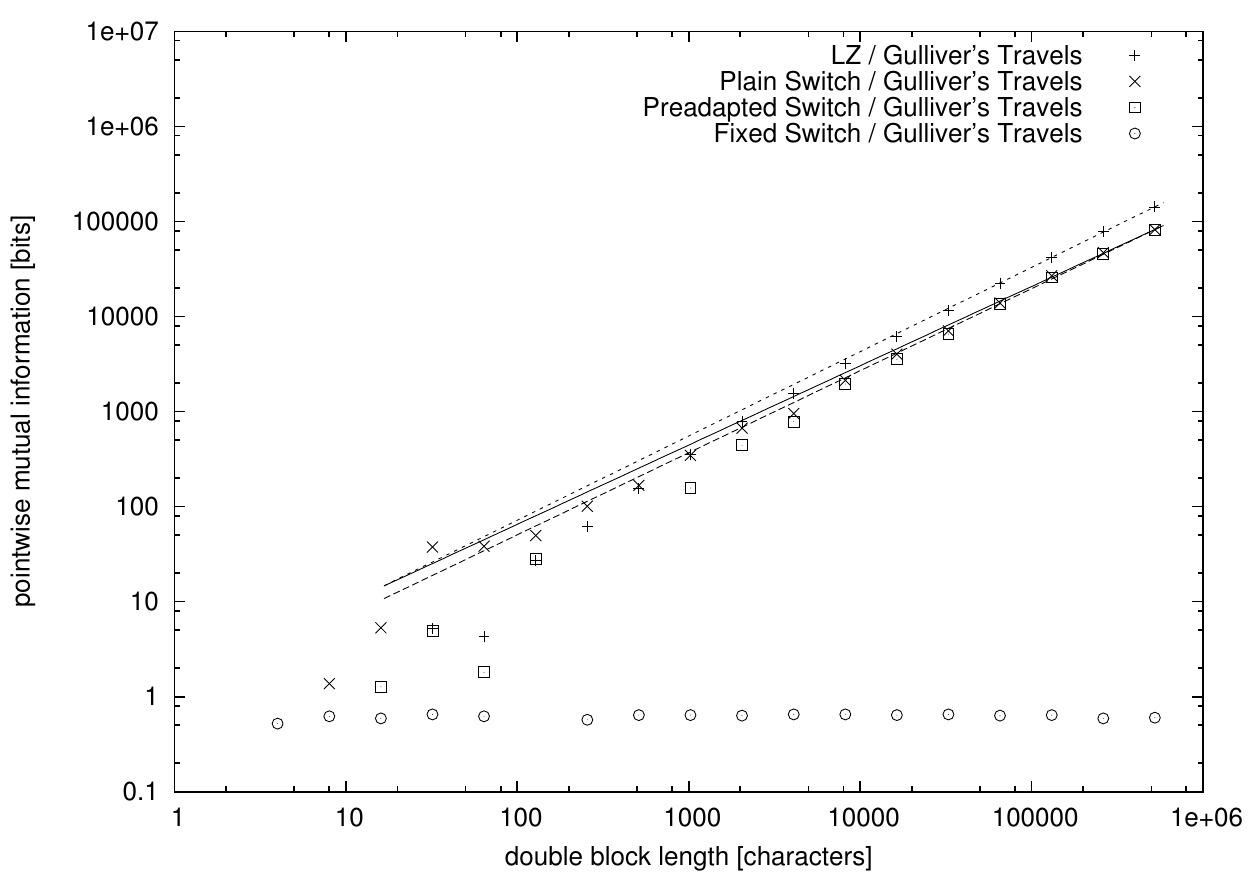}
\caption{\label{figPMI} Pointwise mutual information for the switch
  distributions and the LZ code.  The solid line is the least square
  regression $y=1.395 n^{0.834}$, computed for the plain switch
  distribution. The dashed line is the least square regression
  $y=0.946 n^{0.863}$, computed for the preadapted switch
  distribution. The dotted line is the least square regression
  $y=1.209 n^{0.887}$, computed for the LZ code.}
\end{figure}

\begin{table}
  \centering
\begin{tabular}{|l|l|l|l|l|}
\hline
&\multicolumn{4}{|c|}{$H^P(x_1^n)/n$ [bpc]}\\
\cline{2-5}
$n$&LZ&plain switch&preadapted switch&fixed switch\\
\hline
2       &7.7459  &8.3547  &6.4605  &7.2124  \\  
4       &10.2089 &8.6367  &6.2363  &6.2506  \\  
8       &10.1347 &8.4657  &5.847   &5.5963  \\   
16      &10.1122 &8.1227  &5.8676  &5.6687  \\ 
32      &9.9482  &6.9604  &5.1395  &5.0712  \\   
64      &9.8816  &6.8478  &4.9319  &4.8489  \\
128     &9.4894  &6.7355  &5.254   &5.8753  \\ 
256     &9.1817  &5.9023  &4.9481  &5.2167  \\ 
512     &8.8427  &5.2381  &4.7506  &4.8141  \\ 
1024    &8.5069  &4.6831  &4.414   &4.626   \\ 
2048    &8.0525  &4.1411  &4.0127  &4.4325  \\ 
4096    &7.7158  &3.9809  &3.8476  &4.4953  \\ 
8192    &7.3084  &3.6209  &3.5361  &4.4023  \\ 
16384   &6.9471  &3.3941  &3.3238  &4.3935  \\ 
32768   &6.5467  &3.0459  &3.0114  &4.3422  \\
65536   &6.1909  &2.7745  &2.7504  &4.327   \\
131072  &5.865   &2.5342  &2.5223  &4.3188  \\
262144  &5.5665  &2.3759  &2.3664  &4.3142  \\
524288  &5.2928  &2.2252  &2.213   &4.3126  \\
\hline
\end{tabular}
\caption{\label{tabCompressionRates} Compression rates 
  for the switch distributions and the LZ code.}
\end{table}

\begin{table}
  \centering
\begin{tabular}{|l|l|l|l|l|}
\hline
&\multicolumn{4}{|c|}{$I^P(x_1^{n/2};x_{n/2+1}^{n})$ [bits]}\\
\cline{2-5}
$n$&LZ&plain switch&preadapted switch&fixed switch\\
\hline
2       &-1.32      &-0.71          &-0.64           &-2.14\\  
4       &-6.04      &-1.13          &-0.67           &0.52\\  
8       &-3.99      &1.37           &-0.5            &0.62\\   
16      &-4.96      &5.3            &1.27            &0.59\\ 
32      &5.25       &37.5           &4.96            &0.65\\   
64      &4.26       &38.26          &1.8             &0.62\\
128     &27.27      &49.6           &28.22           &-2.47\\ 
256     &61.92      &100.78         &-3.89           &0.57\\ 
512     &155.1      &166.76         &-29.5           &0.64\\ 
1024    &353.89     &345.54         &156.62          &0.64\\ 
2048    &789.56     &668.01         &441.52          &0.63\\ 
4096    &1554.31    &954.54         &786.68          &0.65\\ 
8192    &3187.28    &2128.53        &1945.46         &0.65\\ 
16384   &6119.95    &4017.21        &3558.77         &0.64\\ 
32768   &11608.28   &7062.68        &6551.43         &0.65\\
65536   &22241.83   &13877.79       &13549.91        &0.63\\
131072  &41621.75   &26852.4        &25727.25        &0.64\\
262144  &78530.25   &47113.91       &45313.69        &0.59\\
524288  &142330.87  &81859.85       &81873.95        &0.6 \\
\hline
\end{tabular}
\caption{\label{tabPMI} Pointwise mutual information
  for the switch distributions and the LZ
  code.}
\end{table}

In Figure \ref{figCompressionRates} and Table
\ref{tabCompressionRates}, the quality of compression can be compared
for the particular distributions. Among the universal schemes, the
best compression is given by the preadapted switch distribution
followed by the plain switch distribution followed by the LZ code.
However, our hope that the preadapted switch distribution will
significantly beat the plain switch distribution has not been
fully confirmed. Indeed for short blocks the preadapted switch distribution
mimics the behavior of the fixed switch distribution and performs much
better than the plain switch distribution. Alas, for long blocks the
difference between the two universal switch distributions becomes
negligible. Ultimately, both universal switch distributions compress
the text twice better than the LZ code. For no universal code we can
observe the ultimate stabilization of the compression rate. On the
other hand, the fixed switch distribution, which is not universal,
stabilizes ultimately at the constant rate of 4.31 bpc.

The stabilization of the fixed switch distribution is clearly visible
in Figure \ref{figPMI} and Table \ref{tabPMI}, which concern the
pointwise mutual information. Namely, we can see that pointwise mutual
information for the fixed switch distribution does not grow, whereas
for the other distributions, which are universal, the pointwise mutual
information grows rather fast. The tightest bound for the pointwise
mutual information is obtained in the case of the plain switch
distribution, which gives the exponent $\gamma= 0.83$ for the codewise
Hilberg conjecture (\ref{HilbergP}). It is surprising that the
pointwise mutual information for the two other universal distributions
grows almost at the same rate, despite the large difference of
compression rates between the universal switch distributions and the
LZ code. 

As we have indicated in the Introduction, observing $\gamma> 0$ for
the switch distribution, we have to reject the hypothesis that the
source is a Markov chain. Possibly, Hilberg's conjecture
(\ref{Hilberg}) may be true but we need still some stronger evidence
for this hypothesis such as a lower bound for the exponent $\beta$.

\section*{Acknowledgment}

The author wishes to thank Jacek Koronacki and Jan Mielniczuk for a
discussion and an anonymous referee for suggesting some relevant
references.

\bibliographystyle{IEEEtran}

\bibliography{0-journals-abbrv,0-publishers-abbrv,ai,mine,tcs,nlp,ql,books}

\begin{thebibliography}{10}
\providecommand{\url}[1]{#1}
\csname url@rmstyle\endcsname
\providecommand{\newblock}{\relax}
\providecommand{\bibinfo}[2]{#2}
\providecommand\BIBentrySTDinterwordspacing{\spaceskip=0pt\relax}
\providecommand\BIBentryALTinterwordstretchfactor{4}
\providecommand\BIBentryALTinterwordspacing{\spaceskip=\fontdimen2\font plus
\BIBentryALTinterwordstretchfactor\fontdimen3\font minus
  \fontdimen4\font\relax}
\providecommand\BIBforeignlanguage[2]{{%
\expandafter\ifx\csname l@#1\endcsname\relax
\typeout{** WARNING: IEEEtran.bst: No hyphenation pattern has been}%
\typeout{** loaded for the language `#1'. Using the pattern for}%
\typeout{** the default language instead.}%
\else
\language=\csname l@#1\endcsname
\fi
#2}}

\bibitem{Hilberg90}
W.~Hilberg, ``{Der bekannte Grenzwert der redundanzfreien Information in Texten
  --- eine Fehlinterpretation der Shannonschen Experimente?}'' \emph{Frequenz},
  vol.~44, pp. 243--248, 1990.

\bibitem{EbelingNicolis91}
W.~Ebeling and G.~Nicolis, ``Entropy of symbolic sequences: the role of
  correlations,'' \emph{Europhys.\ Lett.}, vol.~14, pp. 191--196, 1991.

\bibitem{EbelingNicolis92}
------, ``Word frequency and entropy of symbolic sequences: a dynamical
  perspective,'' \emph{Chaos Sol.\ Fract.}, vol.~2, pp. 635--650, 1992.

\bibitem{EbelingPoschel94}
W.~Ebeling and T.~P\"oschel, ``Entropy and long-range correlations in literary
  {E}nglish,'' \emph{Europhys.\ Lett.}, vol.~26, pp. 241--246, 1994.

\bibitem{BialekNemenmanTishby01}
W.~Bialek, I.~Nemenman, and N.~Tishby, ``Predictability, complexity and
  learning,'' \emph{Neural Comput.}, vol.~13, p. 2409, 2001.

\bibitem{BialekNemenmanTishby01b}
------, ``Complexity through nonextensivity,'' \emph{Physica A}, vol. 302, pp.
  89--99, 2001.

\bibitem{CrutchfieldFeldman03}
J.~P. Crutchfield and D.~P. Feldman, ``Regularities unseen, randomness
  observed: {The} entropy convergence hierarchy,'' \emph{Chaos}, vol.~15, pp.
  25--54, 2003.

\bibitem{Debowski11b}
{\L}.~D\k{e}bowski, ``On the vocabulary of grammar-based codes and the logical
  consistency of texts,'' \emph{IEEE Trans.\ Inform.\ Theory}, vol.~57, pp.
  4589--4599, 2011.

\bibitem{Debowski12}
------, ``Mixing, ergodic, and nonergodic processes with rapidly growing
  information between blocks,'' \emph{IEEE Trans.\ Inform.\ Theory}, vol.~58,
  pp. 3392--3401, 2012.

\bibitem{Debowski06}
------, ``On {Hilberg}'s law and its links with {Guiraud}'s law,'' \emph{J.\
  Quantit.\ Linguist.}, vol.~13, pp. 81--109, 2006.

\bibitem{Herdan64}
G.~Herdan, \emph{Quantitative Linguistics}.\hskip 1em plus 0.5em minus
  0.4em\relax Butterworths, 1964.

\bibitem{Heaps78}
H.~S. Heaps, \emph{Information Retrieval---Computational and Theoretical
  Aspects}.\hskip 1em plus 0.5em minus 0.4em\relax Academic Press, 1978.

\bibitem{Shannon51}
C.~Shannon, ``Prediction and entropy of printed {English},'' \emph{Bell Syst.\
  Tech.\ J.}, vol.~30, pp. 50--64, 1951.

\bibitem{Debowski13}
{\L}.~D\k{e}bowski, ``Empirical evidence for {Hilberg's} conjecture in
  single-author texts,'' in \emph{Methods and Applications of Quantitative
  Linguistics---Selected papers of the 8th International Conference on
  Quantitative Linguistics (QUALICO)}, I.~Obradovi\'c, E.~Kelih, and
  R.~K\"ohler, Eds.\hskip 1em plus 0.5em minus 0.4em\relax Belgrade: Academic
  Mind, 2013, pp. 143--151.

\bibitem{ZivLempel77}
J.~Ziv and A.~Lempel, ``A universal algorithm for sequential data
  compression,'' \emph{IEEE Trans.\ Inform.\ Theory}, vol.~23, pp. 337--343,
  1977.

\bibitem{Ryabko10}
B.~Ryabko, ``Applications of universal source coding to statistical analysis of
  time series,'' in \emph{Selected Topics in Information and Coding Theory},
  ser. Series on Coding and Cryptology, I.~Woungang, S.~Misra, and S.~C. Misra,
  Eds.\hskip 1em plus 0.5em minus 0.4em\relax World Scientific Publishing,
  2010.

\bibitem{GrayDavisson74b}
R.~M. Gray and L.~D. Davisson, ``The ergodic decomposition of stationary
  discrete random processses,'' \emph{IEEE Trans.\ Inform.\ Theory}, vol.~20,
  pp. 625--636, 1974.

\bibitem{Kallenberg97}
O.~Kallenberg, \emph{Foundations of Modern Probability}.\hskip 1em plus 0.5em
  minus 0.4em\relax Springer, 1997.

\bibitem{Weissman05}
T.~Weissman, ``Not all universal source codes are pointwise universal,'' 2004,
  \url{http://web.stanford.edu/~tsachy/pdf_files/Not%20All%20Universal%20Source%20Codes%20are%20Pointwise%20Universal.pdf}.

\bibitem{BrownOthers92}
P.~F. Brown, S.~A.~D. Pietra, V.~J.~D. Pietra, J.~C. Lai, and R.~L. Mercer,
  ``An estimate of an upper bound for the entropy of {English},''
  \emph{Comput.\ Linguist.}, vol.~18, pp. 31--40, 1992.

\bibitem{Jelinek97}
F.~Jelinek, \emph{Statistical Methods for Speech Recognition}.\hskip 1em plus
  0.5em minus 0.4em\relax The MIT Press, 1997.

\bibitem{ManningSchutze99}
C.~D. Manning and H.~Sch\"utze, \emph{Foundations of Statistical Natural
  Language Processing}.\hskip 1em plus 0.5em minus 0.4em\relax The MIT Press,
  1999.

\bibitem{ClearyWitten84}
J.~G. Cleary and I.~H. Witten, ``Data compression using adaptive coding and
  partial string matching,'' \emph{IEEE Trans.\ Comm.}, vol.~32, pp. 396--402,
  1984.

\bibitem{WillemsShtarkovTjalkens95}
F.~M.~J. Willems, Y.~M. Shtarkov, and T.~J. Tjalkens, ``The context tree
  weighting method: Basic properties,'' \emph{IEEE Trans.\ Inform.\ Theory},
  vol.~41, pp. 653--664, 1995.

\bibitem{RonSingerTishby96}
D.~Ron, Y.~Singer, and N.~Tishby, ``The power of amnesia: Learning
  probabilistic automata with variable memory length,'' \emph{Machine Learn.},
  vol.~25, pp. 117--149, 1996.

\bibitem{KiefferYang00}
J.~C. Kieffer and E.~Yang, ``Grammar-based codes: {A} new class of universal
  lossless source codes,'' \emph{IEEE Trans.\ Inform.\ Theory}, vol.~46, pp.
  737--754, 2000.

\bibitem{SalomonMotta}
D.~S. an~Giovanni~Motta, \emph{Handbook of Data Compression}.\hskip 1em plus
  0.5em minus 0.4em\relax Springer, 2009.

\bibitem{ErvenGrunwaldRooij07}
T.~van Erven, P.~Gr\"unwald, and S.~de~Rooij, ``Catching up faster in
  {Bayesian} model selection and model averaging,'' in \emph{Advances in Neural
  Information Processing Systems 20 (NIPS 2007)}, 2007.

\bibitem{KollenRooij13}
W.~M. Koolen and S.~de~Rooij, ``Universal codes from switching strategies,''
  \emph{IEEE Trans.\ Inform.\ Theory}, vol.~59, pp. 7168--7185, 2013.

\bibitem{HindleRooth93}
D.~Hindle and M.~Rooth, ``Structural ambiguity and lexical relations,''
  \emph{Comput.\ Linguist.}, vol.~19, pp. 103--120, 1993.

\bibitem{LouchardSzpankowski97}
G.~Louchard and W.~Szpankowski, ``On the average redundancy rate of the
  {Lempel-Ziv} code,'' \emph{IEEE Trans.\ Inform.\ Theory}, vol.~43, pp. 2--8,
  1997.

\bibitem{Grunwald07}
P.~D. Gr\"unwald, \emph{The Minimum Description Length Principle}.\hskip 1em
  plus 0.5em minus 0.4em\relax The MIT Press, 2007.

\bibitem{Atteson99}
K.~Atteson, ``The asymptotic redundancy of {Bayes} rules for {Markov} chains,''
  \emph{IEEE Trans.\ Inform.\ Theory}, vol.~45, pp. 2104--2109, 1999.

\bibitem{Shields97}
P.~C. Shields, ``String matching bounds via coding,'' \emph{Ann.\ Probab.},
  vol.~25, pp. 329--336, 1997.

\bibitem{Debowski12b}
{\L}.~D\k{e}bowski, ``Maximal lengths of repeat in {English} prose,'' in
  \emph{Synergetic Linguistics. Text and Language as Dynamic System},
  S.~Naumann, P.~Grzybek, R.~Vulanovi\'c, and G.~Altmann, Eds.\hskip 1em plus
  0.5em minus 0.4em\relax Wien: Praesens Verlag, 2012, pp. 23--30.

\bibitem{Ukkonen95}
E.~Ukkonen, ``On-line construction of suffix trees,'' \emph{Algorithmica},
  vol.~14, pp. 249--260, 1995.

\end{thebibliography}

\end{document}